\documentclass[journal,onecolumn]{IEEEtran}
\usepackage{times}
\usepackage{graphicx}
\usepackage{amssymb, amsmath,amsthm, amsfonts}
\usepackage{ifthen}
\usepackage{bbm}
\usepackage{tikz}
\usepackage{enumerate}
\usepackage{verbatim}

\newtheorem{theorem}{Theorem}

\newtheorem{lemma}[theorem]{Lemma}
\newtheorem{prop}[theorem]{Proposition}

\newtheorem{definition}[theorem]{Definition}

\newcommand{\be}{\begin{equation}}
\newcommand{\ee}{\end{equation}}
\newcommand{\ben}{\begin{equation*}}
\newcommand{\een}{\end{equation*}}
\newcommand{\ba}{\begin{eqnarray}}
\newcommand{\ea}{\end{eqnarray}}

\newcommand{\indicator}[1]{\mathbbm{1}_{\{ {#1} \} }}

\def\hgi{|h_i|^2}

\title{Power Controlled Adaptive Sum-Capacity of  Fading MACs with Distributed  CSI}
\author{
\IEEEauthorblockN{Sibi Raj~B.~Pillai, Bikash~K.~Dey, Yash~Deshpande and Krishnamoorthy~Iyer} \\
\IEEEauthorblockA{Department of Electrical Engineering \\Indian Institute of Technology Bombay. \\
	{\tt \{bsraj,bikash\}@ee.iitb.ac.in}
}

}

\begin{document}

\maketitle

\usetikzlibrary{arrows}
\begin{abstract}
We consider the problem of finding optimal, fair and distributed power-rate strategies to achieve
 the sum capacity of the Gaussian multiple-access block-fading channel. In here, the transmitters have
 access to only their own fading coefficients, while the receiver has global access to all
 the fading coefficients.
Outage is not permitted in any communication block. The resulting average
sum-throughput is also known as `power-controlled adaptive sum-capacity', which appears
as an open problem in literature. 

This paper presents the power-controlled adaptive sum-capacity of a wide-class of popular 
MAC models. In particular, we propose a power-rate strategy in the presence of distributed channel
state information (CSI), which is throughput optimal when all the users have
identical channel statistics. The proposed scheme also has an efficient implementation using
successive cancellation and rate-splitting. 
We propose an upperbound when the channel laws are not 
identical. Furthermore, the optimal schemes are extended to situations in which each
transmitter has additional finite-rate partial CSI on the link quality of others. 
\end{abstract}

\section{Introduction}
\label{sec:intro}
 The multiple access channel (MAC) is a widely studied model in information theory, where
many users communicate to a single entity using a shared medium. With its natural
 applications in wireless communications, the so called fading MAC with additive 
white Gaussian noise is one of the popular MAC models. In here, the channel
from each user to the receiver is modeled by a multiplicative fading channel.

In order to find the rate-tuples at which reliable communication is possible
over the fading MAC model, it is important to make assumptions about the amount of
channel knowledge available at the transmitters and the receiver. It is natural
to assume that the receiver has access to the fading coefficients, by means
of pilot symbols. In other words, the receiver has full CSI.
 On the other hand, the same is not true about the transmitter.
We consider a model where each transmitter is fully aware of its own fading
coefficient (individual CSI), but that of no other. Towards the latter sections of this paper, we 
relax this assumption and equip the transmitter with finite-rate partial CSI 
of other links. 

We consider a slow fading model, which is modeled by block fading: the fading 
coefficients are constant over a
 block of channel uses, over which the codeword lasts. The transmitters, thus, are not allowed
 to take advantage of the ergodic nature of the fading process during coding, but may employ
 adaptive power and rates. This particular situation is motivated by systems involving occasional
 (opportunistic) access to a shared medium, such as in a cognitive radio or a sensor network with
 a star topology. Here, multiple users wish to communicate their data to the receiver over the
 awarded time slot in a fair but distributed fashion. These systems may lack the global
 user coordination information to do conventional multiplexing strategies like TDM. However,
 some limited coordination information can be made available or gleaned from the network, 
 for example, the total number of active users participating in a given slot.
%
%
It is natural to look for within-block coding in these systems and demand that communication
in each block to be outage-free, while
allowing for adaptively controlling the power and transmission-rates based on the available
channel knowledge. Furthermore, the power-adaptation strategy should respect the corresponding
average  transmit power at each  of the users.

 There is considerable literature on  multiaccess fading channels with instantaneous CSI. The
 Shannon capacity of a Gaussian MAC with CSI available only at the receiver is evaluated rigorously
 in \cite{ShamaiWyner97}. The optimal power control strategies to achieve capacity for the case of
 complete channel state information at the transmitters (CSIT) are given in \cite{KnoppHumblet95}
 and \cite{TseHanly98}. Coming to partial side information at the transmitters, \cite{DasNarayan02}
 gives the capacity region of a fading MAC under very general notions of CSI at the transmitters.
 These notions can be specialized to nearly all practical scenarios including individual transmitter
 CSI. However, our work differs from \cite{DasNarayan02} due to the block-fading assumption and the
 requirement of no outage in each block. The
 ergodic averaging inherently used in evaluating the Shannon capacity region in
 \cite{DasNarayan02} turns out to be essential because of the absence of complete CSI. Alternate
 notions of capacity motivated by different practical scenarios have also been  investigated:
 delay-limited capacity for the fading MAC is dealt with in \cite{TseHanly98} while
 \cite{EffrosGoldsmith10} defines the notions of expected capacity and capacity with outage for
 information unstable single-user channels. Other related works which consider partial CSI
 in a fading MAC setup are \cite{CemalSteinberg05}, where non-causal CSI is considered,
 and its generalization and unification with causal CSI in \cite{Jafar06}. 

The model that we consider in this paper, i.e. block fading MAC with individual CSI, 
is applicable to  a wide-range of situations. The model can effectively capture
a MAC with random access, 
where the availability
of a packet for transmission is indicated by binary fading states, available at
the respective transmitters~\cite{QiBe03}, \cite{AdTo05}, \cite{HwSeCi06}.
Building on this, \cite{HwMaGaCi07} generalized the fading MAC model and 
enforced an additional requirement that 
the probability of error approaches zero for every state realization, or block.
Clearly, the underlying assumption of a sufficiently large block-length allows
the construction of capacity achieving coding strategies, with error-probabilities
vanishing in the blocklength. An important utility in this case is the throughput, which is
the average data-rate over blocks. A single letter characterization in terms of the
block-wise MAC capacity regions is provided in \cite{HwMaGaCi07} for the
DM-MAC and fading AWGN MAC.
Averaging the possible communication rates over different state realizations will
result in the  adaptive capacity region in the presence of distributed CSI. In a Gaussian  
setting, the users can also adapt their power in addition to the rates \cite{GamalKim11}.
The maximal sum-throughput in this case is named as
the   power controlled adaptive sum-capacity (see Section~23.5.2 of \cite{GamalKim11}), 
which is  the maximal empirical average of the sum-rates achieved in each block.
In \cite{HwMaGaCi07}, \cite{GamalKim11}, it is also shown that the general 
power-controlled adaptive capacity region corresponds to 
a convex problem, and the sum-capacity is numerically determined for two-state
fading  MACs using convex programming techniques. The evaluation 
of adaptive sum-capacity for the popular Rayleigh fading is an open 
problem \cite{HwMaGaCi07}. 
We present  a solution to this problem in the current paper, 
which is almost closed form
(in terms of a water-filling formula), along with
several other interesting results. Recently, other approaches connecting random access
and MAC with distributed CSI has been reported~\cite{MiFrTs12}. 

Another related work considers rateless coding in a distributed MAC set-up \cite{Urs06}.
 Here, a set of users try to convey their respective message, but
 the slot (block) duration is not fixed. Each user is unaware of the link quality of other users or
even the number of users in the system. A communication round ends when all the messages
are successfully recovered at the receiver. The receiver employs a feedback beacon signal for 
synchronizing the rounds of communication. Interestingly, it is demonstrated that the distributed nature
of access results in negligible loss of throughput efficiency. As opposed to this, here we consider 
fixed slot-duration, and there is no feedback signal to the transmitters. Nevertheless,
we are still able to construct a  rate-splitting scheme, paving the way for  low complexity
throughput-optimal strategies, albeit in conjunction with an unconventional decoding
architecture.

For most parts of the paper we  consider identical channel statistics across users.
This is a reasonable assumption in many practical systems. The strategies that
we propose  have a simpler structure when the average powers are also the same 
across users.  We identify this special case by the name \emph{identical users}. More precisely,
throughout the paper, the usage \emph{identical users} is synonymous with the
following two constraints. 
 \begin{itemize} 
 \item the fading statistics are iid across users.
 \item each user has the same average power constraint.
 \end{itemize} 
%
While we state the results for arbitrary power constraints at the transmitters, we will first 
single out the identical users case, for the ease of exposition.

Our main results are summarized as follows:
\begin{enumerate}
 \item We introduce a fair, simple and distributed policy called the `alpha-midpoint' strategy for the
 Gaussian multiple-access block-fading channel. 
\item The alpha-midpoint strategy achieves the power controlled adaptive sum-capacity when the
	channel statistics are identical.
\item We  propose a low-complexity rate-splitting scheme that allows the alpha-midpoint strategy
 to be implemented through successive decoding.
\item  For identical users, we compute the power-controlled adaptive sum-capacity when each 
user is additionally given some symmetric finite-rate partial CSI of the other links. 
\item For non-identical channel statistics, an upper bound to the power-controlled adaptive 
sum-capacity is presented. 
\end{enumerate}

%
%
%


The organization of the paper is as follows. Section~\ref{sec:defn} will introduce 
the system model and some notations, and also  defines the notion of
\textbf{power-controlled adaptive sum-capacity}, which is our utility of interest. 
In Section~\ref{sec:midpoint} we introduce a communication strategy utilizing
the available individual CSI, called
the \emph{midpoint strategy}. To implement the communication strategy with 
low complexity successive cancellation, we propose a rate-splitting strategy in 
Section~\ref{sec:ratesplit},
along with an unconventional decoding architecture.
The power-controlled adaptive sum-capacity is evaluated in 
Section~\ref{sec:uneq}.
Section~\ref{sec:psi} extends our results to the case where
additional partial finite-rate side information on the other links is available at 
the transmitters. Bounds for non-identical channels are detailed in 
Section~\ref{sec:nonid}. We also describe some interesting connection between our model
and the so called L-out-of-K MAC (LOOK channel), this is done in Section~\ref{sec:connect}.  
Section~\ref{sec:conc} concludes the paper with a discussion of results and
possible extensions.   
\section{System Model and  Definitions\label{sec:defn}}
\label{sec:system}
Consider $L$ users communicating with a single receiver. 
These users transmit real-valued signals $X_i$,
 encountering real-valued fades $H_i$. If $Y$ is the value of the received signal at a (discrete) time
 instant we have
\begin{equation} \label{eq:one}
Y = \sum_{i}^L H_i X_i + Z
\end{equation}
 where $Z$ is an independent Gaussian noise process. The fading space
 $\mathcal{H}_i$ of the $i$-th user is the set of values taken by $H_i$,
 and the joint fading space $\mathcal{H}$ is the set of values taken by
 the joint fading state $\bar{H} = (H_1, H_2, \cdots, H_L)$.
 Similar vector quantities of user-wise parameters, like rate, power, channel
 state realization, will be denoted similarly, i.e., with an overbar symbol.

 We consider a slow-fading model, where each channel coefficient stays constant
within a block and varies across blocks in an i.i.d fashion. While we demand
reliable communication within each block, the utility that
we consider is the average sum-throughput, or average sum-capacity, where
the average is over different fading realizations or blocks. A more precise
definition of our utility is given later in this section. In Gaussian channels,
rate-expressions usually take a logarithmic form, and in this paper, all
logarithms are expressed to the base of $2$, that the rates we talk about 
are expressed as bits.

 
We assume that
  the (stationary and ergodic) fading processes $H_i$ are independent,
 and their distributions are known to all the transmitters and the receiver.
  In addition, we have \emph{individual} CSIT, i.e. each transmitter knows
 its own channel fading coefficient $H_i$ but that of no other. The receiver
 knows all the fading coefficients. The transmitters have individual
 average power constraints. i.e.
\begin{align}
\int_h P_i(h)d\Psi_i(h) \leq P_i^{avg},\,\, 1 \leq i \leq L,
\end{align}
where $\Psi_i(\cdot)$ is the cumulative channel law (cdf) of user~$i$.
The users can
 adapt the rate (and power) according to their own channel conditions.
 Apart from the notation changes, our model and objectives are similar in spirit to
 the those 
 presented in \cite{GamalKim11} (see Section~23.5), \cite{HwMaGaCi07}. In fact, 
 the terminology \emph{power-controlled adaptive sum-capacity} is borrowed from 
 these references, which we explicitly compute in the present paper.

 The adaptive nature of communication naturally leads to  the following notion 
 of a power-rate strategy.
\begin{definition}
 A power-rate strategy is a collection of mappings $(P_i, R_i) : \mathcal{H}_i \longmapsto 
  \mathbb{R}^+ \times \mathbb{R}^+$; $i=1,2,\cdots, L$.  Thus, in the 
 fading state $H_i$, the $i^{\textrm{th}}$ user expends power $P_i(H_i)$ and employs a codebook of
 rate $R_i(H_i)$. 
\end{definition}

 Let $C_{MAC}(\bar{h},\bar{P}(\bar{h}))$ denote the capacity region  of a  Gaussian 
 multiple-access channel   with fixed channel gains of $\bar h= h_1, \cdots, h_L$ and 
 respective power allocations  $\bar{P}(\bar{h})=(P_1(h_1), \cdots , P_L(h_L))$. We know that, 
\begin{align}
C_{MAC} (\bar{h},\bar{P}(\bar{h})) = \Biggl\{ \bar{R} \in \{\mathbb R^+\}^L : \forall S 
		\subseteq \{1,2, \cdots, L\},\, 
	 \sum_{i \in S} R_i \leq \frac{1}{2} \log 
 		\left( 1 + \sum_{i \in S} \hgi P_i (h_i) \right) \Biggr\}
\end{align}

\begin{definition}
 We call a power-rate strategy as \emph{feasible} if it satisfies the average power 
constraints for each user i.e.\\ $\forall i \in \{1, 2, \cdots, L\}, 
	\quad \mathbb{E}_{{H}_i} P_i(H_i) \leq P_i^{avg}$.  
\end{definition}
\begin{definition}
A power-rate strategy is termed as \emph{outage-free} if it never results in outage i.e. 
\begin{align*}
\forall \bar{h} \in \mathcal{H},  (R_1(h_1), \cdots, R_L(h_L)) 
		\in C_{MAC}(\bar h ,\bar P (\bar h))
\end{align*}
\end{definition}
Notice that the definition of outage-free implicitly makes use of the fact that
the blocklength  is sufficiently large for achievable strategies to reach acceptable
error-probabilities.  While a very large or infinite block-length may be required to drive
the error vanishingly small, we save on the notation by the simpler definition of 
outage-free as given above. 

Let $\Theta_{MAC}$ be the collection of all feasible power-rate strategies which are
outage-free. 
%
%
Let us now specialize the definitions to the case of identical channel statistics,
i.e. the cdf of each user is $\Psi(h)$.
For any strategy $\theta \in \Theta_{MAC}$, the throughput is
\begin{align} 
T_{\theta}  =       \sum_{i=1}^L \mathbb E R^{\theta}_i(H_i)  
 &= \sum_{i=1}^L \int_h R^{\theta}_i(h) \,  d\Psi(h) \nonumber \\  
 &=  \int_h d\Psi(h) \, \left(\sum_{i=1}^L R^{\theta}_i(h)\right) , \label{eq:tput:defn} 
\end{align}
where the superscript $\theta$ is used to identify the feasible power-rate strategy 
employed. i.e. $R^{\theta}_i(h)$ is the rate allocated to user $i$ while observing
fading coefficient $h$. The corresponding transmit power is denoted as $P^{\theta}_i(h)$.


\begin{definition} \label{def:sum:capa}
The \textbf{power controlled adaptive sum-capacity} is the maximum (average) throughput 
achievable, i.e. $C_{sum}(\Psi) =  \max_{\theta \in \Theta_{MAC}} T_{\theta}$.
\end{definition}
One of the main results of the paper is the computation of the power-controlled 
adaptive sum-capacity for several popular fading models. In the special case of
a single user channel ($L=1$), the adaptive sum-capacity is well known, as it
becomes a full CSI model. We denote the single user-capacity with an average power
constraint of $P_a$  as $C_1(\Psi,P_a)$, which can be evaluated using a water-filling formula
\cite{TseViswanath05}, 
\begin{align} \label{eq:tput:su}
C_1(\Psi,P_a)  = \frac 12 \int d\Psi(h) \log(1 + |h|^2 P^*(h)) ,
\end{align}
where 
\begin{align}\label{eq:tput:su:pc}
P^*(h) =  \left( \frac 1{\lambda} - \frac 1{|h|^2} \right)^+
 \text{ and }
\int d \Psi(h) P^*(h) = P_a.
\end{align}
The single user water-filling formula is considered to be closed form for all
practical purposes. Our results  for the MAC with distributed CSI also take
the form of similar water-filling formulas. Thus, we will re-use the notation
$C_1(\Psi, P_a)$ several times in this paper.  
Let us now focus on computing $C_{sum}(\Psi)$ for $L>1$.

\section{Power Controlled Adaptive Sum-capacity ($C_{sum}(\Psi)$)}
 Throughout this section, we consider identical fading statistics across 
 all users of a $L-$user MAC with distributed CSI. 
The main result
 of this section is  to compute the power-controlled adaptive sum-capacity
 $C_{sum}(\Psi)$ for arbitrary $\Psi(\cdot)$.
 We first
 state the main result and then explain its structure and implications, before
 providing the proof.   
\begin{theorem} \label{thm:mp:opt}
Given independent and identically distributed channels according to the  c.d.f $\Psi(h)$,
\begin{align} \label{eq:tput:alpha}
C_{sum}(\Psi)  = C_1(\Psi, \sum_{i=1}^L P_i^{avg}).
\end{align}
\end{theorem}
\label{sec:midpoint}
Before proving this, it is instructive to observe the structure of the result. 
The result states that
$C_{sum}(\Psi)$ is same as the capacity of a single user channel with cdf $\Psi(\cdot)$
and average power $\sum_{i=1}^L P_i^{avg}$.  It has an element of surprise in the
first look, as if there is some degeneracy in the problem statement. While this
is not the case, the single user result essentially comes from the fact that 
communication has to be outage-free in every block. The result is re-stating that
the worst joint distribution of the MAC fading-states is a highly correlated one,
in which the same fading value is observed across users. 

We propose a   distributed strategy which achieves $C_{sum}(\Psi)$,  termed as  
\textbf{alpha mid-point strategy} in this paper. In order to clearly present the ideas
behind this scheme, we first consider a special case of the model. The proof of 
Theorem~\ref{thm:mp:opt} 
is detailed  in Section~\ref{sec:uneq}. The special case that we consider now is
that of \emph{identical users}, i.e. users having identical channel statistics
and similar average powers.

\subsection{Identical Users and Mid-point Strategy}
 For identical users with individual CSI, TDMA is a natural scheme for communication.
 Let us first compute the achievable rates for TDMA and then consider possible
alternatives.   
\subsubsection{Plain TDMA}
In plain TDMA the transmitters employ a simple `taking turns' policy. Each block is divided into
 sub-blocks with only one user transmitting in that sub-block. This requires some extra
 coordination such as agreeing on an ordering for the users. The channels for the users are now
 \emph{orthogonal} and they may water-fill over their own sub-blocks to improve throughput. Thus,
 we obtain the power-rate strategy corresponding to plain TDMA as:
\begin{align*}
P_i(h_i) &= \left( \frac{1}{\lambda_i} - \frac{1}{|h_i|^2} \right)^+ \\
R_i(h_i) &= \frac{1}{2L} \log \left( 1 + L \hgi P_i(h_i) \right)
\end{align*}
where $\lambda_i$ is chosen such that $\mathbb{E}_{{H}_i}P_i(H_i) = P_i^{avg}$. The actual
 power employed by the user in its sub-block is $LP_i(H_i)$ and the full transmission rate
 supported thereby is chosen.
\subsubsection{The Midpoint Rate Strategy}
We now present an appealing alternative to TDMA, which also has some advantages over TDMA.
This strategy is  distributed in nature, and we call it the \emph{midpoint} strategy.
\begin{figure}
\begin{center}
\begin{tikzpicture}[>=stealth,scale=1.0]
\draw[->] (0,0) --++(0,4.75) node[left]{$R_1$};
\draw[->] (0,0) --++(5,0) node[below]{$R_2$};
\draw[red, line width=1.5pt] (0,4) --++(2,0) --++(2,-2) --++(0,-2);
\draw[blue, line width=1.5pt] (0,1) --++(0.5,0) --++(0.5,-0.5) --++(0,-0.5);
\draw[gray, line width=1.5pt] (0,1) --++(3.25,0) --++(0.75,-0.75) --++(0,-0.25);
\draw[dashed] (0.75,0.75) --++(2.25,0) node[below]{$A$} --++(0,2.25);
\end{tikzpicture}
\caption{The users $1$ and $2$ construct the innermost and outermost   MAC capacity regions
 respectively. The intermediate pentagon is the instantiated(actual) MAC region and $A$ denotes the
 operating point\label{fig:midptexp}}.
\end{center}
\end{figure}
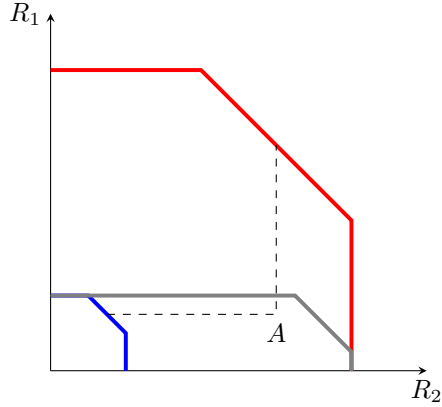
Consider a block  in which the  fading vector is $h_1, h_2, \cdots, h_L$ and let
the  respective powers be $P_1(h_1), P_2(h_2), \cdots, P_L(h_L)$
(as part of  some  feasible power
 strategy). Each user assumes that all others are identical to itself (in terms of fading
coefficients and transmit powers) and constructs the
 symmetric MAC region based on this assumption. It then chooses the maximal equal-rates point
 for  operation. Thus we have
\begin{equation} \label{eq:mid:rate}
R^{mid}_i(h_i) = \frac{1}{2L} \log \left( 1 + L \hgi P_i(h_i) \right).
\end{equation}
\begin{lemma}
The midpoint rate strategy is outage-free, i.e.
\begin{equation*}
\forall \bar{h} \quad \bar{R}_i^{mid} \in C_{MAC}(\bar{h}, \bar{P}).
\end{equation*}
\end{lemma}
\begin{proof}
The lemma follows directly from the concavity of the logarithm function, i.e.
 $\forall S \subset \{1, 2, \cdots, L\}$, 
\begin{align*}
\quad \sum_{i \in S}
         R^{mid}_i(h_i)  &= \frac{1}{2L} \sum_{i\in S} \log \left( 1 + L  \hgi P_i(h_1) \right)\\
	&\leq \frac{1}{2|S|} \sum_{i\in S} \log \left( 1 + |S| \hgi P_i(h_i) \right)\\
	&\leq \frac{1}{2} \log \left( 1 + \sum_{i \in S} \hgi P_i(h_i) \right).
\end{align*}

\end{proof}
As $P_1, \cdots, P_L$ are arbitrary, the users' power strategies can now be
decoupled. The best power strategy for each user would thus be to
 water-fill over its own channel and we obtain
\begin{equation} \label{eq:pow:alloc}
P_i(h_i) = \left( \frac{1}{\lambda_i} - \frac{1}{|h_i|^2} \right)^+
\end{equation}
where $\lambda_i$ is chosen such that $\mathbb{E}_{H_i}P_i(H_i) = P_i^{avg}$.
%
In addition to its apparent simplicity, it also turns out that  the \emph{midpoint} strategies 
achieve the power-controlled adaptive sum-capacity
for identical users - a special case of Theorem~\ref{thm:mp:opt}. 

 Notice that the throughput achieved by the midpoint strategy is identical to that achieved
 by plain TDMA. We compare this in Figure \ref{fig:PCcomparison} with the optimal opportunistic-TDMA
 (O-TDMA)
 possible in the presence of complete CSIT \cite{KnoppHumblet95}. The figure shows the
throughputs of the two schemes for a normalized Rayleigh fading channel, corrupted with
AWGN of unit gain. The users are assumed to have identical average powers. 
 The advantage of plain 
 TDMA over the midpoint strategy is its simplicity in decoding, since
 only $L$ single-user decoders are needed. However, the price for this is paid in the extra
 coordination required to set up an ordering for transmission between users. The midpoint
 strategy avoids this coordination, albeit at the cost of incurring joint decoding. We show
 in the next section that this cost can be ameliorated through rate-splitting and
 successive decoding.
\begin{figure}
\begin{center}
\includegraphics[scale=0.7]{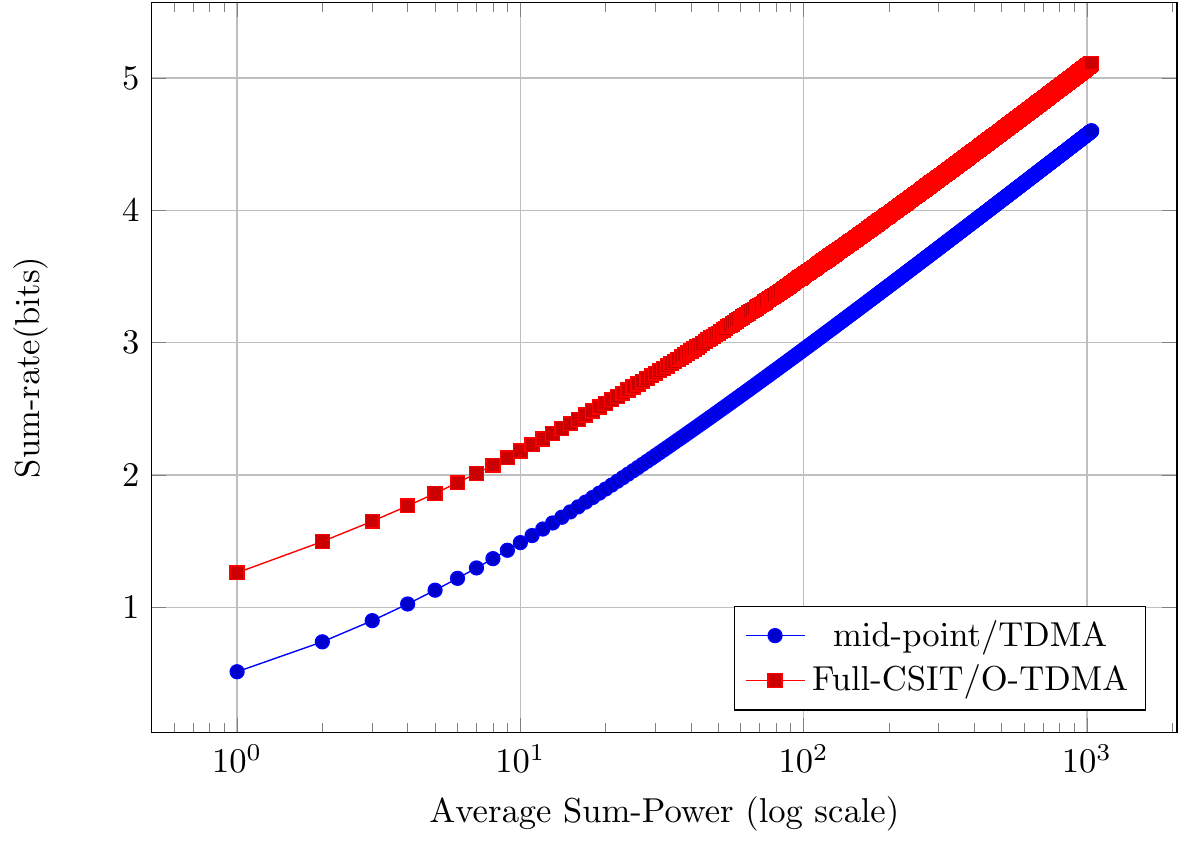}
\caption{The midpoint strategy is only a constant off the full CSI bound \cite{KnoppHumblet95}
 \label{fig:PCcomparison}}
\end{center}
\end{figure}

\subsection{Unequal Average Powers} \label{sec:uneq}
We now prove Theorem~\ref{thm:mp:opt} in two steps.  First, 
we construct an upperbound to $C_{sum(\Psi)}$. The second step generalizes the
mid-point strategy to construct an achievable scheme which meets the proposed
upper-bound.
\subsubsection{\underline{An Upperbound}}
Let us consider a single link with cdf $\Psi(h)$, let  $\Theta_s(\tilde P)$ be the collection of all 
single-user power  allocation  schemes ($P_1(h)$) such that
\begin{align}
\int P_1(h) d\Psi(h) = \tilde P.
\end{align} 
Let $P_{sum} = \sum_{i=1}^L P_1^{avg}$, also recall the definition of throughput in \eqref{eq:tput:defn}.
\begin{lemma}  \label{lem:up:1} The throughput $T_{\theta}$ obeys,
$$T_{\theta} \leq C_{1}(\Psi,P_{sum}), \forall \theta \in \Theta_{MAC}.$$
\end{lemma}
\begin{proof}
\vspace*{-0.25cm}
\begin{align} 
 T_{\theta} &\overset{(a)}{\leq} \frac 12 \int_h d\Psi (h)\, \log \left(1+|h|^2\sum_{i=1}^L 
	P^{\theta}_i(h)\right) \label{eq:up:bound0} \\
 &\leq \max_{\Theta_s(\sum P_i^{avg})}  \frac 12 \int d\Psi(h) \, \log \left(1+|h|^2 P(h)\right). 
	\label{eq:up:bound}
\end{align}
Here $(a)$ follows from \eqref{eq:tput:defn}, by applying the sum-rate upper bound on a 
MAC with  received signal power $\sum_i |h|^2 P^{\theta}_i(h)$.
The second inequality results from relaxing the individual power constraints to
 a single average sum-power constraint.  

It is clear that water-filling of the inverse fading gains is the optimal strategy in 
a point to point fading channel under an average power constraint. Thus the last
expression above is indeed $C_{1}(\Psi, P_{sum})$.
\end{proof}
 
\subsubsection{\underline{Alpha-midpoint strategy}}
The \textbf{alpha-midpoint strategy} is a generalization of the midpoint scheme that
we introduced earlier. 
Let $\bar \alpha$ be a vector of
non-negative values with $\sum_i \alpha_i = 1$.
In alpha-midpoint strategy, the rate chosen by user $i$ while encountering
a fading coefficient of $h_i$ is, 
\begin{align}\label{eq:alpha:achieve}
R^{\bar{\alpha}}_i(h_i) = \alpha_i \frac 12  \log \left( 1 
		+ |h_i|^2\frac{P_i(h_i)}{\alpha_i} \right) ,
\end{align}  
where $P_i(h_i)$ is the transmitted power, chosen such that
$$
\int P_i(h) d\Psi(h) = P_i^{avg}.
$$

\begin{lemma}
The alpha-midpoint strategy is outage-free.
\end{lemma}
\begin{IEEEproof}
For any $S\subseteq \{1,2,\cdots, L\}$,
\begin{align}
\sum_{i\in S} R^{\bar{\alpha}}_i(h_i) &= \sum_{i\in S} \alpha_i \frac 12  
		\log \left( 1 + |h_i|^2\frac{P_i(h_i)}{\alpha_i} \right) \\
	&\leq \frac 12 \log\left( 1 + \sum_{i\in S} |h_i|^2{P_i(h_i)} \right),
\end{align}
by concavity of the logarithm. Clearly the chosen rate-tuple across users
is within $C_{MAC}(\bar h, \bar P(\bar h))$ for every block, ensuring that
there is no outage.
\end{IEEEproof}
We now show the optimality of alpha-midpoint schemes.
\begin{lemma} \label{lem:in:2}
\begin{align*}
\max_{\theta \in \Theta_{MAC}} T_{\theta} = C_1(\Psi, P_{sum})
\end{align*}
\end{lemma} 
\begin{IEEEproof}
We will specialize our alpha-midpoint strategy to achieve $C_{1}(\Psi, P_{sum})$. 
To this end, choose for $1 \leq i \leq L$,
\begin{align} \label{eq:alpha:val}
\alpha_i = \frac{P_i^{avg}}{\sum_{i=1}^L P_i^{avg}} \text{ and } 
P_i(h) = \alpha_i P^*(h),
\end{align}
where $P^*(h)$ is given in \eqref{eq:tput:su:pc}, with $P_a$ replaced by $P_{sum}$. 
Notice that,
\begin{align*}
\int P_i(h) d\Psi(h) &= \alpha_i \int P^*(h) d\Psi(h) \\
		&= \alpha_i P_{sum} \\
		&= P_i^{avg}.
\end{align*}
Furthermore, by \eqref{eq:alpha:achieve}
\begin{align*}
\sum_{i=1}^L \int R^{\bar{\alpha}}_i(h_i) d\Psi(h_i) &= 
	\sum_{i=1}^L \frac{\alpha_i}2 \int \log(1+|h|^2\frac{P_i(h_i)}{\alpha})d\Psi(h) \\
	&= \sum_{i=1}^L \frac{\alpha_i}2 \int \log(1+|h|^2P^*(h))d\Psi(h) \\
	&= \frac 12 \int \log(1+|h|^2P^*(h))d\Psi(h)\left( \sum_{i=1}^L \alpha_i\right) \\
	&= C_1(\Psi, P_{sum})
\end{align*}
This completes the proof of Theorem~\ref{thm:mp:opt}.
\end{IEEEproof}


%

\section{Rate Splitting and Successive Decoding}
\label{sec:ratesplit}
The alpha mid-point strategies have a fairly simple structure. In this section, we show that
these strategies can be implemented by low complexity successive decoding architectures. 
To this end, 
we present an asymptotically optimal rate-splitting strategy that mitigates the requirement
 of joint decoding, replacing the joint decoder with $L N_v$ successive single-user decoders,
 where $N_v$ is a parameter, signifying the number of layers per user.  
 This section is motivated by the work in \cite{Urs06} and their
 technique is useful in showing the achievability. However, \cite{Urs06} considers a rateless
 scheme with variable coding block-lengths between rounds of  communication. The length of
 each round is determined by a feedback link from the receiver, which announces the
 next round via a beacon. Our scenario requires that the communication occur within a fixed
 block or time slot, and there is no assumption of such a feedback link.

 We will first construct rate-splitting schemes for \emph{identical users}. 
Extensions to arbitrary average powers is done in a separate subsection.
We will write the received signal power for user~$i$,
 i.e. $P_i |h_i|^2$ as simply $\gamma_i$, throughout this section.
 Assume that the users have different (received) powers
 $\gamma_1, \gamma_2, \cdots,  \gamma_{L}$. For simplicity, we will assume that the additive noise is of unit
 variance.  The values of $\gamma_i$ may change with each block of communication depending on the
 individual fading conditions. Each user is \emph{unaware}  of the fade values and
 transmit powers of the
 rest of the users and, consequently, the interference they may cause.

The encoding and decoding are done as: each user splits itself into $L$ virtual
users and splits its power, perhaps unequally, among these users. Each user can be
visualized as a `stack' of virtual users.
For decoding, we use a successive cancellation based single-user decoder,
which decodes one of the virtual users assuming all other virtual users which
are not yet decoded as Gaussian noise, see \cite{RimoldiUrbanke96} for the details.
More specifically, user $i$, having received power $\gamma_i$, splits its data stream in to $N_v$ virtual
users. This is done by allotting a power/rate pair $(\gamma_i^l,r_i^l)$ to the $l^{th}$ virtual user,
such that $\sum_l \gamma_i^l = \gamma_i$.
 The transmitter $i$ assumes that all other users are also at (received)
 power $\gamma_i$ and imagines identical power splitting strategies across all users. It
 then chooses  the rates $r_i^l$ by considering all the other virtual users in the same
 and lower layers as interference, i.e.,
 \vspace*{-0.0cm}
\begin{align} \label{eq:rate:split}
 r_i^l = \frac{1}{2} \log \left( 1 + \frac{\gamma_i^l}{1 + (L-1)\gamma_i^l  + 
L\sum_{j = 1}^{l-1} \gamma_i^j } \right).
\end{align}
%
\noindent However, in the actual setting, the interference encountered from  the other users are
substantially different from that accounted for in the denominator of \eqref{eq:rate:split}.
So a layer by layer decoding may fail, as the virtual users are not chosen according to the actual
channel conditions.  Surprisingly, it turns out that this can be compensated by not strictly
 adhering to a layer by layer decoding. In particular,
the receiver retains the freedom to decode the topmost hitherto undecoded
layer of \emph{any} transmitter,
irrespective of the number of layers which were already decoded. It is, in fact,
this freedom that allows the transmitters to choose the virtual rates without knowledge of
interference from the other users.

\begin{lemma}
 Assuming layer-wise rate allocation as per (\ref{eq:rate:split}),
it is always possible to find a virtual user
which can be decoded correctly, i.e. with arbitrarily small error probability.
\end{lemma}
\begin{proof}
We prove this by induction. Assume that layers (virtual users) above $l_k$ have been decoded for the
$k^{\mathrm{th}}$ transmitter. Choose:
\begin{align*}
\kappa = \arg\max_{k} \sum_{j = 1}^{l_k} \gamma_k^j
\end{align*}
For user $\kappa$, the remaining interference for decoding layer $l_{\kappa}$:
\begin{align*}
1 +  \sum_{j=1}^{l_{\kappa}-1} \gamma_{\kappa}^j &+ \sum_{k \neq \kappa} \sum_{j = 1}^{l_k} \gamma_k^j  \\
&= 1 + \sum_{k =1}^{L} \sum_{j = 1}^{l_k} \gamma_k^j - \gamma_{\kappa}^{l_{\kappa}} \\
&\leq 1 + L \sum_{j = 1}^{l_{\kappa}} \gamma_{\kappa}^j - \gamma_{\kappa}^{l_{\kappa}} \\
&= 1+ \sum_{j=1}^{l_{\kappa}-1} \gamma_{\kappa}^j +  (L-1)\sum_{j=1}^{l_{\kappa}} \gamma_{\kappa}^j
\end{align*}
The inequality follows directly from the choice of $\kappa$. The RHS is the \emph{expected} interference
for the $l_{\kappa}$th virtual user of transmitter $\kappa$. Thus, as the actual interference is less
than the expected interference, this virtual user can be correctly decoded. In other words, the user
 with the `best' received SNR can always be chosen for decoding.
\end{proof}
\begin{theorem}
As $N_v \rightarrow \infty$ and $\forall j, l, \quad \gamma_l^j \rightarrow 0$ , the rate achieved by all the
users approach their midpoint rate.
\end{theorem}
\begin{IEEEproof}
Using \ref{eq:rate:split}, we have
\begin{align*}
R_i = \sum_{j = 1}^{N_v} \frac{1}{2} \log \left( 1 + \frac{\gamma_i^l}{1 + (L-1)\gamma_i^l  
                + L\sum_{j = 1}^{l-1} \gamma_i^j } \right)
\end{align*}

Under the given conditions, we can use the same method as in Lemma 1 of \cite{Urs06}  to show that:
\begin{align*}
\lim_{N_v \rightarrow \infty} R_i &= \lim_{N_v \rightarrow \infty} \sum_{j = 1}^{N_v} \frac{\gamma_i^l}{1 +
                 (L-1)\gamma_i^l  + L\sum_{j = 1}^{l-1} \gamma_i^j } \\
&= \frac{1}{2}\int_{0}^{\gamma_i} \frac{dy}{1+Ly} = \frac{1}{2L} \log \left( 1 + L \gamma_i \right)
\end{align*}
\end{IEEEproof}
Computational results in \cite{Urs06}  show that only a moderate number of virtual users $N_v$
         are sufficient to yield good performance. For completeness, the increase in achievable-rate
with respect to the number of virtual layers is shown in Figure~\ref{fig:layers:capa}.
\begin{figure}[htbp]
\begin{center}
\includegraphics[scale=0.9]{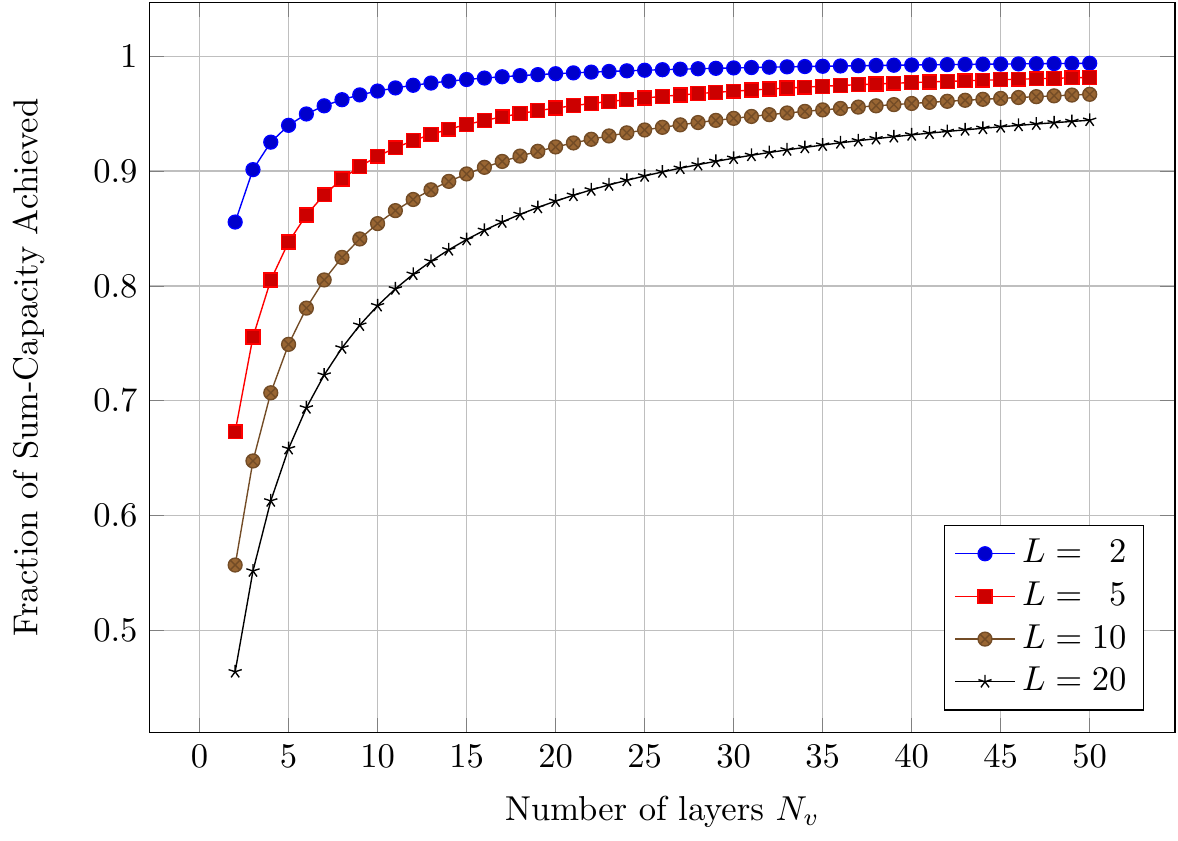}
\end{center}
\caption{The fractional sum-capacity achieved Vs  $N_v$\label{fig:layers:capa}}
\end{figure}

\subsection{Unequal Average Powers}
We will construct two levels of splitting in the presence of unequal average powers. 
In particular, we first split each user $k$ into  $N_k$ virtual users, in such a way
that  each virtual user has an identical\footnote{sometimes the transmit-power of 
users may not be
commensurate, however we can choose a \emph{slightly} lower power level for some of the
users, with negligible loss of performance.} average transmit power constraint of $P_v$, 
irrespective of the user index $k$. Thus,
$$
\sum_{i=1}^{N_k} P_v =  P_k^{avg}.
$$
Evaluating the maximal average rate for the $L^{\prime} = \sum_{k=1}^L N_k$ virtual users
under the midpoint strategy of \eqref{eq:mid:rate} will also yield $C_{1}(\Psi, P_{sum})$. 
To see this,  notice that the total-rate obtained by the $N_k$ layers of user~$k$ is
$$
N_k \frac1{L^{\prime}} \int \frac 12  \log(1 + |h|^2P^*(h)) d\Psi(h),
$$
where $P^*(h)$ is the single-user water-filling allocation with an average power of
$L^{\prime} P_v = \sum_{k=1}^L P_i^{avg} = P_{sum}$. Notice that $N_k/L^{\prime}$
is nothing but the $\alpha_k$ in \eqref{eq:alpha:val}, proving that the above
strategy can achieve the same rates as the alpha-midpoint scheme. 
Furthermore, since the midpoint rates are  achievable by single user decoding
techniques~\cite{DPD11}, alpha midpoint rates can also be achieved by low complexity
schemes.

\section{Finite-rate CSI on Other Links\label{sec:psi}}
Up to this point, we have assumed  only individual CSI. In this section,  we 
wish to study the effect of additional partial information about  the other links. 
To keep things simple, we consider
\emph{identical users} with the specified cdf $\Psi(h)$ and an average power of $P^{avg}$. 
Extensions to unequal average power constraints are possible, but not covered here. To start with,
we consider $1$ bit of additional partial CSI, i.e., each transmitter gets one bit of information from 
every other link, in addition to its own individual CSI. The individual fading components
are assumed to be independently chosen. The additional link-information 
given to other transmitters  is only a function of the
fading parameter of this link. Thus, the model captures situations where  the extra bit is
obtained through transmitter cooperation or cribbing. It is crucial that the receiver has no 
say on  the partial CSI. If the receiver decides the conveyed bit, then the throughput  is 
same as that of the full CSI~\cite{KnoppHumblet95}. 

The partial CSI contains link quality information: let us assume it to be chosen from the set
 $\{G,B\}$, where
we used $G$ for good and $B$ for bad. A natural separation between $G$ and $B$  is a link gain 
threshold. In particular, the partial CSI bit 
$\hat h_k$ of transmitter $k$ is
\begin{align}
\hat h_k = \begin{cases} 
	G \text{ if } |h_k| \geq h_{T} \\
	B \text{ otherwise,}
	\end{cases}
\end{align}
for some fixed positive threshold $h_T$. By slight abuse of notation, we will say that
link $j$ is in state $G$ (and call it good user), and denote the probability of that 
event by $\mu(G)$. Using the same
token, $1- \mu(G) = \mu(B)$. 
%
%
Let $C_{PSI} $ be the maximum attainable
throughput  with $1$ bit additional CSI on each of the other links, along with individual CSI.    
\begin{theorem}\label{thm:psi}
For $L$ identical users,
\begin{align}
C_{PSI} = C_1\left(\Psi^{\prime},LP^{avg}\right),
\end{align}
where the cdf $\Psi^{\prime}(\cdot)$ is such that,
\begin{align*}
d\Psi^{\prime}(h) = d\Psi(h) \left([\mu(B)]^{L-1} \indicator{h\in B}
			+ (1+\zeta)\indicator{h\in G} \right)
\end{align*}
and the parameter
\begin{align}
\zeta =
	\sum_{m=1}^{L-1} {L-1 \choose m} [\mu(B)]^m [\mu(G)]^{L-1-m}
		\frac m{L-m}.
\end{align}
\end{theorem}
\begin{IEEEproof}
Recall the definition of $C_{1}(\cdot, \cdot)$ given in 
\eqref{eq:tput:su}--\eqref{eq:tput:su:pc}. 
We explain the proof for $L=2$, which contains all the essential features.
The proof is relegated to  appendix~\ref{app:psi:proof}. 
\end{IEEEproof}
It is instructive to
compare the advantages of $1$ bit of extra CSI, which we demonstrate for a two user
\textit{identical} Rayleigh fading links of unit second moment, see Figure~\ref{fig:2}. The threshold value $h_T$ for $1-$bit CSI was 
taken as unity. 
\begin{figure}[htbp]
\begin{center}
\includegraphics[scale=1]{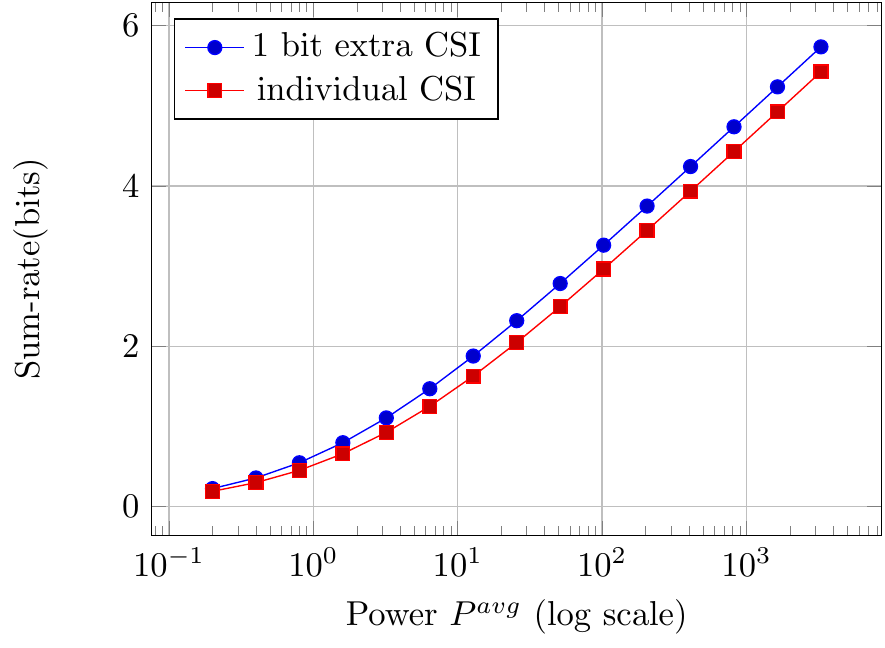}
\caption{Sum-rate improvement by additional CSI\label{fig:2}}
\end{center}
\end{figure}

One immediate question is the sensitivity of the results with respect to the fading threshold. 
We argue that it is not that crucial and the results are robust. In fact,
taking the median of the fading distribution seems to be natural choice for many models.
Numerical results show that even for moderate power-levels the difference from the
 the best choice of the threshold  is barely noticeable. 

Extending to multiple bits of CSI is straightforward when the users are identical and
the additional CSI bits are generated in a symmetric manner, i.e. $n-1$ identical threshold values
define the $\log_2 n$ bits of partial CSI about each link. For each $0\leq m \leq n$, all users who
 experience a link-fading in the threshold bracket 
 $[h_m,h_{m+1}), 0\leq m \leq n$ will form  group $m$, 
where $h_0=0$ and $h_{n+1} = \infty$. An optimal strategy is to let all the  
users in group $m$ transmit at their respective  midpoint rates,
 if there are no users in any group above $m$.

\section{Nonidentical Channel Statistics\label{sec:nonid}}
Our second result is a generalization to non-identical channel statistics. 
In this case, we do not know the optimal schemes, but we
provide an upperbound, which seems to be close for several 
channels of practical interest. W.l.o.g  consider non-negative
valued fading coefficients (by taking modulus), and
let the respective cdf of the individual channels 
be $\{F_1(\cdot), F_2(\cdot), \cdots, F_L(\cdot)\}$. We will assume each of them
to be right continuous and define the corresponding inverse functions as
\begin{align}
\forall \gamma \in [0,1],\, F^{-1}_k (\gamma) = \min h : F_k(h) \geq \gamma.
\end{align}  
For convenience, we will denote $F_k^{-1}(\cdot)$ by $h_k^2(\cdot)$. 
\begin{lemma}\label{lem:up:general}
For non-identical channels defined by the c.d.fs $F_k(\cdot),1\leq \!k \!\leq\! L,$ the 
maximal sum-rate $C_{sum}$ is bounded by
\begin{align} \label{eq:csimac:ub:1}
C_{sum} \leq \max_{\theta_{MAC}} \int_0^1\!\frac 12\!  \log\left( 1
		+\sum_{k=1}^L h_k^2(x) P_k[h_k(x)]\right) dx.
\end{align}
\end{lemma}  
\begin{IEEEproof}
Imagine that the range of each cdf $F_k, 1\leq k\leq L$ in $[0,1]$ is 
divided into $n$ equal segments. Let the inverse map of the $j^{th}$ segment
of cdf $F_k$ be $h^2_k(j/n)$.  The lemma states that for each segment $j$, the
MAC formed by the corresponding inverse maps of this segment should obey the 
sum-rate constraint.  
\end{IEEEproof}
Notice that different channel values are coupled in the above bound (through
their cdf structure), and
we can maximize the power allocation on these coupled fading vectors,
thus obtaining a bound to the RHS of \eqref{eq:csimac:ub:1}.
%
%
By using Lagrange optimization as in \cite{KnoppHumblet95}, we get the following lemma.
\begin{lemma}  \label{lem:up:gen:2}
\begin{align} \label{eq:up:gen:2}
C_{sum} \leq \sum_{k=1}^L \int_0^1 \!\frac 12 \!\log(1 +  h_k^2(x) P_{k}(x)) \alpha_k(x) dx
\end{align}
where 
\begin{align*}
P_k(x) = \left(\frac 1{\lambda_k} - \frac 1{h_k^2(x)} \right)^+
\textrm{ and }
\int_0^1 P_k(x) 
		\alpha_k(x)dx = P_k^{avg}.
\end{align*}
In here, $\alpha_k(x)$ are non-negative functions such that 
$ \sum_{k=1}^L \alpha_k(x)=1, \forall x \in (0,1)$.

\end{lemma}
\begin{IEEEproof}
Maximizing \eqref{eq:csimac:ub:1}  over all coupled channel vectors will yield the bound in
\eqref{eq:up:gen:2}. 
The detailed proof is available in  appendix~\ref{app:non:ident}.  
\end{IEEEproof}

While the existence of $\lambda_k$ and the functions $\alpha_k(x)$
is enough for the proof,   numerical algorithms are required to find
these, except for special cases. One such case where the algorithm is 
straightforward is when the channel coefficients are generated by the same law, but scaled by
different average gains. In this case, $\alpha_k(x) = P_k^{avg}/\sum P_k^{avg}$ for all
$x\in (0,1)$ and
the water-filling formula can be evaluated  using single-user water-filling.
For example, consider a $2-$user Rayleigh faded  MAC, with $E|\textbf h_2|^2 = 
2 E|\textbf h_1|^2 = 2$, and  $P_1^{avg} = P_2^{avg}$. Figure~\ref{fig:one} compares our 
upperbound against the rates  resulting from an adaptation of the 
alpha-midpoint strategy. The lower bound  is 
obtained by considering a symmetric Rayleigh fading channel with $E|h_1|^2 = E|h_2|^2 = 2$
and $\alpha_1 = \frac 13$ and $\alpha_2= \frac 23$. For this strategy, 
the sum-power is taken as $P_1^{avg} + P_2^{avg}$.
\begin{center}
\begin{figure}[htbp]
\begin{center}
\includegraphics[scale=1]{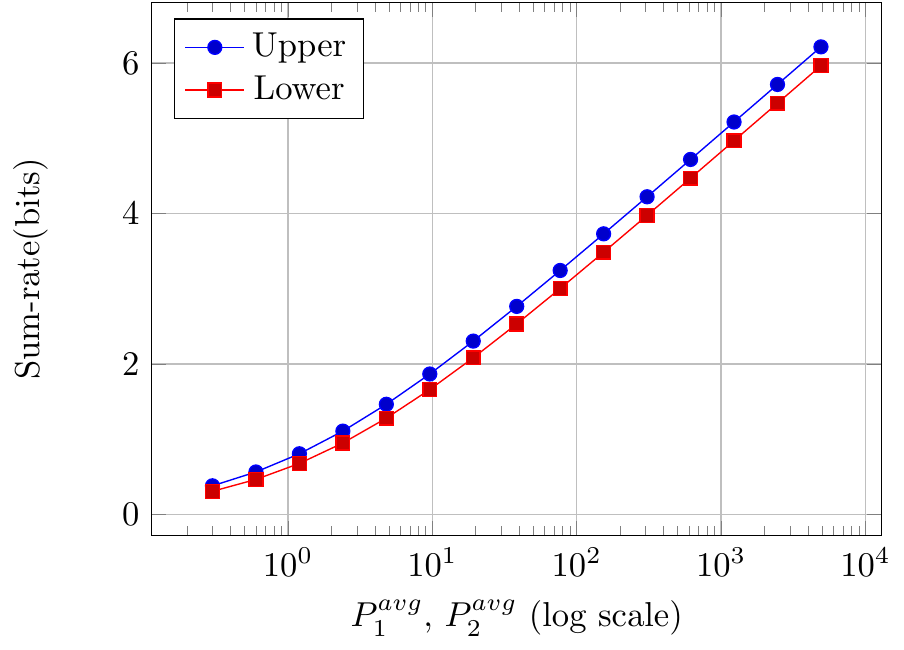}
\caption{Upper and Lower bounds to adaptive sum-capacity\label{fig:one}}
\end{center}
\end{figure}
\end{center}
This not only demonstrates the utility of our upper-bound, but also that 
 the alpha mid-point strategy is a good scheme.


\section{Connections with the LOOK Channel}\label{sec:connect}
We will show  the suitability of our strategies in presence of individual CSI to 
the so called LOOK MAC setting.  
The LOOK channel is a multiple access model which models the variability of the
active user set. In particular, of the $K$ available transmitters, at most $L$ are active in 
a communication epoch. The communication scheme should accommodate any allowed choice
of the active user set. The capacity region of a discrete LOOK MAC model is known,
see \cite{AlsuCheng95} for history and other details. 
Specifically, 
a $K-$ dimensional rate-vector is termed as \emph{achievable},
 if for every subset $S$ of indices having
cardinality at most $L$, the rate-coordinates of $S$  belong  to the
corresponding $|S|-$ user capacity region. Notice that in the discrete memoryless case, 
the capacity region is non-enlarging with $K$, and it can possibly decrease due to the 
additional requirements of the incoming users.
Furthermore, strategies like time-sharing are not viable here, as there is no proper
coordination between the users. In reality, the users may not even know the
identity of each other. Computation of the capacity region for even moderate $K$ seems
infeasible due to the curse of dimensionality. 

We will now consider the Gaussian  LOOK channel and extend it to a block-fading setup.
In a block-fading LOOK channel, each communication epoch corresponds to a block. There
are $K$ users in the system, of which $L$ are chosen uniformly at random. 
Communication has to be outage-free in every block 
or epoch. There is individual CSI at the transmitters and
we consider identical users with power $P^{avg}$ and fading cdf $\Psi(h)$.
The sum throughput $C_{LOOK}$ for this system is
given by, 
\begin{theorem}
\begin{align}
C_{LOOK} = \frac 12 \int \log\left(1 +  L h^2 P(h)\right) d\Psi(h),
\end{align}
where
$$
P(h) = \left(\frac 1{\lambda} - \frac 1{|h|^2} \right)^+ \text{ and } \int P(h) d\Psi(h) = \frac KL P_{avg}.
$$ 
\end{theorem} 
The proof follows in a straightforward manner by each active-user employing its
midpoint rate for communicating, as in \eqref{eq:mid:rate}. We do not repeat the steps here. Notice that for any 
fixed $L$, the throughput goes unbounded 
with $K$, contradictory to what one expect in a discrete memoryless case. 
This is due to the fact that each active user can scale his/her power to make up
for its inactivity period.  

If the users have different average powers, we can adapt the alpha-midpoint strategy
by forming virtual users, each one having the same average power. 
The key parameter which will guarantee throughput optimality is the global knowledge of
the total number of virtual users. Thus, a good design principle for such systems is to
 estimate/access  the total number of virtual users in a block, rather than obtaining 
some limited information about the current received powers. Notice that the number of virual 
users is \textit{largely} a function of the total average power of the active users.


\section{Conclusion\label{sec:conc}}
We have presented throughput optimal outage-free communication schemes for a 
block-fading MAC with identical channel statistics, and distributed (individual) CSI.
While the assumption of symmetric channel statistics is very relevant in several
situations, we are currenly extending our work to take care of asymmetric channel
statistics. Further extensions to MIMO channels are also straightforward. 
The current paper focussed on the adaptive sum-capacity, which actually is a step
towards computing the full  capacity
region, a possible future work.   

In case of  additional partial finite-rate CSI, we computed the sum-capaity under
identical users and symmetric CSI. It will be interesting to relax this assumption
and compare the performance, in terms of non-identical users or assymmetric CSI. 
While we limited our discussion here mostly to the block-coding case, ensuring no outage, it will
also be interesting to see whether we can bridge these results to the ergodic
capacity-achieving schemes.





\begin{appendices}
\section{}\label{app:psi:proof}

\vspace*{-0.25cm}
\noindent \textit{Proof of Theorem~\ref{thm:psi}}:\\
We will show the proof for a $2$ user system for simplicity.
Let $\hat h_i$ denote the CSI communicated from
user $i$ to all others. User~$1$ employs a power of $P_1(h_1,\hat h_2)$
and user~$2$ spends $P_2(\hat h_1, h_2)$. Let $R_1(h_1,\hat h_2)$ and  
$R_2(\hat h_1,h_2)$ be the
respective rates chosen. We can express the average sum-rate as,
\begin{align} \label{eq:rate:psi:parts}
R_1 + & R_2 = \int_{G\times G}  \left( R_1(h_1, \hat h_2) + R_2 (\hat h_1, h_2) \right) d\Psi(h_1, h_2) 
	 + \int_{B\times B}  
		\left( R_1(h_1, \hat h_2) + R_2 (\hat h_1, h_2) \right) d\Psi(h_1, h_2) \notag\\
	&  + 
		\int_{B\times G}  \left( R_1(h_1, \hat h_2) + R_2 (\hat h_1, h_2) \right) d\Psi(h_1, h_2) 
	  + 
		\int_{G\times B}  \left( R_1(h_1, \hat h_2) + R_2 (\hat h_1, h_2) \right) d\Psi(h_1, h_2). 
\end{align}
Consider the first term in the summation of the right hand side. By suitably integrating, it can be
written as a single integral,
\begin{align} \label{eq:rate:gbound}
\mu(G)\int_G (R_1(h,G) +  R_2(G,h)) d\Psi(h) \leq   \frac{\mu(G)}2 \int_G \log\left(1 + h^2 (P_1(h,G) + P_2(G,h))\right) d\Psi(h),
\end{align}
which is the sum-rate bound of the corresponding MAC. Similarly, for the second term, 
\begin{align} \label{eq:rate:bbound}
\mu(B)\int_B (R_1(h,B) +  R_2(B,h)) d\Psi(h)   \leq  \frac{\mu(B)}2 \int_B 
	\log\left(1 + h^2 (P_1(h,B) + P_2(B,h))\right)d\Psi(h).
\end{align}
As for the third and fourth terms,  the information on who has the better channel is readily 
available to both parties here. Let us now consider only those  channel states 
$(h_1, h_2) \in \{(G\times B) \bigcup (B\times G)\}$. Let the
average power expenditure on these channel states be $P_{GB}$. Suppose we
 relax our assumption, and give full CSI to each transmitter whenever  one of the 
links is in state $G$ and the other in $B$. Furthermore,  let us enforce only a average 
sum-power constraint
of $P_{GB}$ in these states. In such a system, only the better user transmits with  an 
appropriate power~\cite{KnoppHumblet95}.  This fact can be utilized along with
\eqref{eq:rate:gbound} and \eqref{eq:rate:bbound} to equivalently write the maximum 
throughput as
\begin{multline} \label{eq:jstar:one}
J^* =  \max 
 \frac{\mu(B)}2  \int_B  
        \log\left(1 + h^2 (P_1(h,B) + P_2(B,h))\right)d\Psi(h)  +
\frac{\mu(G)}2 \int_G \log\left(1 + h^2 (P_1(h,G) + P_2(G,h))\right) d\Psi(h)  \\
	+  
	 \frac{\mu(B)}2 \int_G (\log(1 + h^2 P_1(h,B)) +  \log(1 + h^2 P_2(B,h)) )d\Psi(h),
\end{multline}
where the maximization is over $P_1(\cdot, \cdot)$ and $P_2(\cdot, \cdot)$. Furthermore, 
the original individual power constraint is relaxed to an average sum-power 
constraint of the form,
\begin{multline*}
\mu(B) \int_B (P_1(h,B) + P_2(B,h))d\Psi(h) +  \\ \mu(G) \int_G (P_1(h,G) + P_2(G,h)) d\Psi(h) 
	+   \mu(B) \int_{G} (P_1(h,B) + P_2(B,h)) d\Psi(h) \leq 2P^{avg}.
\end{multline*}
Notice that our integration now is over just one variable. Let us denote,
\begin{align*}
P_B(h) = \frac{P_1(h,B) + P_2(B,h)}2 \text{ and }P_G(h) = \frac{P_1(h,G) + P_2(G,h)}2 .
\end{align*}
By the  concavity of logarithm, the maximization can be bounded in terms of the new
variable as
\begin{multline} \label{eq:jstar:star}
J^* \leq \max \frac{\mu(B)}2  \int_B  
        \log\left(1 + h^2 2 P_B(h) \right)d\Psi(h)  +
\frac{\mu(G)}2 \int_G \log\left(1 + h^2 2 P_G(h)\right) d\Psi(h)  \\
        +  
         \frac{\mu(B)}2 \int_G 2 \log(1 + h^2 P_B(h)) d\Psi(h).
\end{multline}
The power constraint, in the new notation, is
\begin{align}
\mu(B) \int_{B} 2P_B(h) d\Psi(h) +  \mu(G) \int_G 2 P_G(h) d\Psi(h) 
	+   \mu(B) \int_{G} 2P_B(h) d\Psi(h) \leq 2P^{avg}.
\end{align}
Let the RHS of \eqref{eq:jstar:star} be denoted as $J^{**}$.  
Further simplification is possible by treating the variable $h$ as one which
belongs to a single-user channel with appropriate distribution and an
average power of $2P^{avg}$. 
\begin{lemma} For $2$ identical-users with individual cdf  $\Psi(\cdot)$,
the maximal throughput with partial CSI is $C_1(\Psi^{\prime},2P^{avg})$, where
\begin{align}
d\Psi^{\prime}(h) = \begin{cases} 
		d\Psi(h)  \mu(B)  \text{ if } h \in B \\
		d\Psi(h) (1+\mu(B)) \text{ if } h \in G  
		\end{cases}
\end{align}
\end{lemma}
\begin{IEEEproof}
First we show that 
$$
C_1(\Psi^{\prime},2P^{avg}) \geq J^{**}.
$$
For the single user channel $\Psi^{\prime}(h)$, consider two power allocation
 schemes $\hat P$ and $\tilde P$ such that
\begin{align}
\hat P(h) =\begin{cases} 2 P_B(h) , h \in B \\ 2 P_G(h), h \in G \end{cases}
\end{align}
and 
\begin{align}
\tilde P(h) =\begin{cases} 2 P_B(h) , h \in B \\ P_B(h), h \in G \end{cases}.
\end{align}
If we use $\hat P$ for a fraction $\frac{\mu(G)}{1+\mu(B)}$ of the times over $\Psi^{\prime}(h)$,
and $\tilde P$ for the remaining fraction, the throughput is
\begin{multline}
\frac{\mu(B)}2 \int_B  \log(1 + h^2 P_B(h))d\Psi(h) + 
	\frac{1 +\mu(B)}2 \frac{\mu(G)}{1+\mu(B)} \int_G \log(1+h^2 2 P_G(h)) d\Psi(h) \\
	+ \frac{1+\mu(B)}2 \frac{2\mu(B)}{1+\mu(B)} \int_G \log(1+h^2 P_B(h)) d\Psi(h),
\end{multline}
which is indeed $J^{**}$. Notice that an average power constraint of $2P^{avg}$ is
maintained under this allocation. Let us now show  that $C_1(\Psi^{\prime},2P^{avg})$ is indeed
achievable for our MAC with partial CSI model. Let $P^{\prime}(h)$ be the optimal
single-user power allocation for the channel $\Psi^{\prime}(h)$. Consider the following
power allocation in \eqref{eq:jstar:one}.
\begin{align*}
P_1(h,G) &= P_2(h,G) = 0\,,\, \forall h \in B\,;\,\, &P_1(h,B) &= P_2(B,h) = P^{\prime}(h)\,,\, \forall h \in G \\
P_1(h,G) &= P_2(h,G) = \frac{P^{\prime}(h)}2 \,,\,\forall h \in  G\,;\,\, 
	&P_1(h,B) &= P_2(B,h) = \frac{P^{\prime}(h)}2\,,\, \forall h \in B
\end{align*}
The users will choose the midpoint rates whenever both users are  either in $B$ or in $G$. In
other cases, only the better user is active. 
Clearly the power constraints are met and the throughput is indeed $C_1(\Psi^{\prime},2P^{avg})$.
\end{IEEEproof}

For $L>2$ users, if there are $K\geq 1$ links in $G$, only those links with 
$h_k\in G$ will transmit at their respective $K-$ user mid-point rates. On
the other hand, if no links are in $G$, all $L$ users transmit at their respective
$L-$user mid-point rates.
 The power allocation can be effectively determined by single user water-filling
of the cdf $\Psi^{\prime}(h)$ given in Theorem~\ref{thm:psi}.

\section{}\label{app:non:ident}
\noindent \textit{Proof of Lemma~\ref{lem:up:gen:2}:}
Notice that we assume arbitrary channel statistics for the links. 
The following proposition on MAC captures the essential idea behind the result.

\begin{prop}\label{pro:tdm:opt}
For a given \emph{fixed} $L-$user MAC with link gains $h_1, \cdots, h_L$ and
respective average transmit powers $P_1, \cdots, P_L$, the maximal sum-rate
can be achieved by time-sharing.    
\end{prop} 

\begin{IEEEproof}
Let user~$i$ transmit for a fraction of time $\beta_i$
with power $\frac{P_i}{\beta_i}$, at its single user capacity. 
By choosing $\beta_i=\frac{h_i^2P_i}{\sum_k h_k^2 P_k}$ we get,
\begin{align}
\sum R_i = \sum_{i=1}^L \frac{\beta_i}2 \log(1+ \sum_{k=1}^{L} h_k^2 P_k),
\end{align}
which is indeed the MAC sum-rate bound. 
\end{IEEEproof}

Let us now relax the maximization in \eqref{eq:csimac:ub:1}. In particular, 
we replace $P_k[h_k(x)]$ by $P_k(\bar h(x))$, where $\bar h(x)$ is the
global fading vector corresponding to the same c.d.f. value $x$ at each 
transmitter. Thus our relaxed optimization problem is,
\begin{align}
\max \frac 12 \int_0^1 \log(1 + \sum_k h_k^2(x) P_k(\bar h(x))) dx,
\end{align}
such that 
\begin{align}
\int_0^1 P_k(\bar h(x)) dx = P_k^{avg},\, \forall k.
\end{align}
By defining Lagrange multipliers, $\lambda_i, 1\leq i \leq L$, one for each
constraint, we can equivalently maximize the cost $J$, where
\begin{align}
J= \frac 12 \int_0^1 \log(1 + \sum_k h_k^2(x) P_k(\bar h(x))) dx 
- \sum_{k=1}^L \lambda_k \int_0^1 P_k(\bar h(x)) dx.
\end{align}
Taking derivative w.r.t to $P_k(\cdot)$ and applying the boundary conditions 
\begin{align}
\frac 12 \frac{h_k^2(x)}{1+\sum_{k=1}^L h_k^2(x) P_k(\bar h(x))}  
	- \lambda_k \geq 0, \, 1\leq k \leq L,
\end{align}
where the inequality becomes equality for the active user-set (ones which are allocated 
non-zero power at a given value of $x$). 
Therefore, we can conclude that power is allocated to user $j$ only if
$$
\frac{h_j^2(x)}{\lambda_j} \geq \frac{h_i^2(x)}{\lambda_i} , \forall i \neq j.
$$
Let $\zeta(x)$ be the maximum value of $\frac{h_j^2(x)}{\lambda_j}$ over $1\leq j \leq L$.
%
Each active-user will achieve  $\zeta(x)$. However, Proposition~\ref{pro:tdm:opt}
will suggest that the active users can time share and achieve the sum-rate. The  power chosen 
by an active user is 
$$
P_i(h_i(x)) =  \max\{0, \frac 1{\lambda_i} - \frac 1{h_i^2(x)}\}.
$$
The instantaneous received power from active user~$i$ while in its transmitting time-fraction is
$\frac{h_i^2(x)}{\lambda_i} -1$. 
However, the fraction of time given to each active user is dependent on  
the channel-laws and average powers. Thus $\alpha_i(x)$ in \eqref{eq:up:gen:2} is the
time-fraction for the active user~$i$, for a given set of channels  determined by the cdf index $x$.
This concludes the proof of Lemma~\ref{lem:up:gen:2}.


\end{appendices}

\bibliographystyle{IEEEtran}
\bibliography{../biblio/poster}{}
\nocite{}

\section*{ACKNOWLEDGMENTS}
The authors  thank Urs Niesen for providing
leads to the usage of  single user decoding strategies. 

\end{document}